\documentclass[letterpaper, 10pt, conference]{ieeeconf} 

\IEEEoverridecommandlockouts 

\usepackage{graphics}
\usepackage{epsfig}
\usepackage{mathptmx}
\usepackage{setspace}
\usepackage{times}
\usepackage{amsmath}
\usepackage{amssymb}
\usepackage[colorinlistoftodos, textwidth=3cm, shadow]{todonotes}
\usepackage{multirow}
\usepackage{dirtytalk}
\usepackage{setspace}
\renewcommand{\baselinestretch}{0.9}

\newtheorem{remark}{Remark}
\newtheorem{assumption}{Assumption}

\newtheorem{theorem}{Theorem}
\newtheorem{lemma}{Lemma}
\newtheorem{definition}{Definition}

\title{\LARGE \bf
Detection of Network and Sensor Cyber-Attacks in Platoons of Cooperative Autonomous Vehicles: a Sliding-Mode Observer Approach
}

\author{Twan Keijzer$^{1}$ and Riccardo M.G. Ferrari$^{1}$
\thanks{$^{1}$Twan Keijzer and Riccardo M.G. Ferrari are with Delft Centre for Systems and Control,
        Delft University of Technology, 2628 CD Delft, The Netherlands
        {\tt\small \{t.keijzer,r.ferrari\}@tudelft.nl}}%
}
\begin{document}
\maketitle
\thispagestyle{empty}
\pagestyle{empty}

\begin{abstract}
Platoons of autonomous vehicles are being investigated as a way to increase road capacity and fuel efficiency. Cooperative Adaptive Cruise Control (CACC) is an approach to achieve such platoons, in which vehicles collaborate using wireless communication. While this collaboration improves performance, it also makes the vehicles vulnerable to cyber-attacks. In this paper the performance of a sliding mode observer (SMO) based approach to cyber-attack detection is analysed, considering simultaneous attacks on the communication and local sensors. To this end, the considered cyber-attacks are divided into three classes for which relevant theoretical properties are proven. Furthermore, the harm that attacks within each of these classes can do to the system while avoiding detection is analysed based on simulation examples.
\end{abstract}

\section{Introduction}\label{sec:intro}
\noindent Cyber Physical Systems (CPS), such as cooperative adaptive cruise control (CACC) platoons, have been getting increased interest over the last decade. Their connectivity is often highly beneficial to their performance. However, it also makes them inherently vulnerable to cyber-attacks, which might form an even bigger threat to their safe operation. 

Detection of cyber-attacks has therefore become an important point of focus in research about CPS. Many cyber-attack detection approaches have since been proposed, of which overviews can be found in \cite{Ding2018,Dibaji2019,Sun2018,Ferrari2021-xs}. 

An important factor driving research on cyber-attack detection is the malicious intent of the attacker: using the available resources s/he will attempt to craft an attack harming the system while avoiding detection. In \cite{Teixeira2015}, the attackers ability to harm the system while remaining stealthy is analysed based on the resources available to them.

In previous work by the authors \cite{Jahanshahi_2018aa,Keijzer2019} a Sliding Mode Observer (SMO) based cyber-attack detection scheme for CACC platoons was proposed. The cyber-attacks considered in that work, however, were not crafted with malicious intent. Therefore, in this paper, the performance of this existing SMO-based detection scheme is analysed for the case of malicious attackers.

\noindent To this end, the paper presents the following contributions:
\begin{itemize}
    \item The existing SMO-based detection scheme is applied to a scenario where a CACC platoon is subjected to simultaneous attacks on the communication and all local measurements.
    \item The considered attacks are divided into three classes, and estimation and detection capabilities within these classes are proven analytically.
    \item One malicious attack is designed within each class with the intent to cause a crash in the platoon while staying undetected. The extend to which this is possible is analysed based on simulation results.
\end{itemize}
\section{CACC Platoon Model}\label{sec:problem}
\noindent The paper considers cyber-attacks on a CACC platoon and aims to determine how attacks can be crafted that both are harmful and avoids detection. The attack detection scheme described in section \ref{sec:smo} will be used against a broad range of attacks on the vehicle communication and measurements.

\subsection{Collaborative Adaptive Cruise Control Model}
\noindent Each car in the platoon is modeled as
\begin{equation}\label{eq:car_model}
    \left\{
\begin{aligned}
\dot{\tilde{x}}_i = \left[\begin{matrix}\dot{p}_i\\\dot{v}_i\\\dot{a}_i\end{matrix}\right]&=\underset{A_i}{\underbrace{\left[\begin{matrix}0&1&0\\0&0&1\\0&0&-\frac{1}{\tau_i}\end{matrix}\right]}}\underset{\tilde{x}_i}{\underbrace{\left[\begin{matrix}p_i\\v_i\\a_i\end{matrix}\right]}}+\underset{B_i}{\underbrace{\left[\begin{matrix}0\\0\\\frac{1}{\tau_i}\end{matrix}\right]}}u_i\,,\\
\tilde{y}_i &= \left[\begin{matrix}p_i-p_{i-1}-L_{i}\\v_i-v_{i-1}\\v_i\\a_i\end{matrix}\right]+\tilde{\zeta}_i\,,
 \end{aligned}\right.
\end{equation}
where subscripts $i$ and $i-1$ denote cars $i$ and $i-1$ respectively. Furthermore, $p$, $v$, $a$, $u$, $\tilde{y}$, $\tilde{\zeta}$, $\tau$, and $L$ are, respectively, the position, velocity, acceleration, control input, measurements, measurement noise, engine time constant, and length of the car. It is assumed that the measurement noise $\tilde{\zeta}$ is bounded and zero-mean.

In the considered scenario, as can be seen from equation \eqref{eq:car_model}, each car measures the distance and relative velocity with respect to the preceding car, as well as its own velocity and acceleration. Furthermore, each car receives, via wireless communication, the control input $u_{i-1}$ of the preceding car. Using this information the controller \cite{Ploeg2011a}, is implemented as\footnote{It is important to remark that the proposed estimation and detection scheme works regardless of the used controller} 
\begin{equation*}
    \dot{u}_i = -\frac{1}{h_\text{ref}}u_i+k_p e_i + k_d \dot{e}_i + \frac{1}{h_\text{ref}}u_{i-1}
\end{equation*}
where $e_i = \tilde{y}_{i,a,(1)}-(r+h_\text{ref}\tilde{y}_{i,a,(3)})$, $r$ and $h_\text{ref}$ are the desired standstill distance and time headway, and $k_p$ and $k_d$ are gains.

\begin{assumption}\label{ass:sys_uncertainty}
Each car $i$ has perfect knowledge only of its own dynamics, i.e. $\tau_{j}$ is known only for $j = i$. The relation $\hat{\tau}_{j}=r_\tau \tau_{j}$ will be used as an estimate of $\tau_{j}$ if $j\neq i$. Here $r_\tau$ is unknown, but its upper and lower bounds are known.
\end{assumption}

The interaction between two cars in a CACC scenario can be modeled, from the perspective of car $i$, as
\begin{equation}\label{eq:CACC_model}
    \left\{
\begin{aligned}
\hspace{-0.1cm}\left[\begin{matrix}\dot{\tilde{x}}_{i-1}\\\dot{\tilde{x}}_i\end{matrix}\right] &= \underset{\hat{A}}{\underbrace{\left[\begin{matrix}\hat{A}_{i-1}&0\\0&A_i\end{matrix}\right]}}\underset{\tilde{x}}{\underbrace{\left[\begin{matrix}\tilde{x}_{i-1}\\\tilde{x}_i\end{matrix}\right]}}+\underset{\hat{B}}{\underbrace{\left[\begin{matrix}\hat{B}_{i-1}&0\\0&B_i\end{matrix}\right]}}\underset{u}{\underbrace{\left[\begin{matrix}u_{i-1}\\u_i\end{matrix}\right]}}+E\tilde{\eta}\\
 \tilde{y}_i &= 
 C\underset{\tilde{x}}{\underbrace{\left[\begin{matrix}\tilde{x}_{i-1}\\\tilde{x}_i\end{matrix}\right]}}+\underset{\tilde{c}}{\underbrace{\left[\begin{matrix}-L_{i}\\0_{3\times 1}\end{matrix}\right]}}+\tilde{\zeta}_i\,,
 \end{aligned}\right.
\end{equation}
where $C$ can be derived from equation \eqref{eq:car_model}. The model uncertainty is made explicit by using $\hat{A}_{i-1}$ and $\hat{B}_{i-1}$, which denote $A_{i-1}$ and $B_{i-1}$ where $\hat{\tau}_{i-1}$ replaces $\tau_{i-1}$. Furthermore,
\begin{equation*}
    E=\left[\begin{matrix}
0_{2\times 1}\\\frac{1}{\hat{\tau}_{i-1}}\\0_{3\times 1}
\end{matrix}\right]\,; \; \tilde{\eta}=(r_\tau -1)(u_{i-1}-a_{i-1})\,,
\end{equation*}
such that $\hat{A}$, $\hat{B}$ and $E$ are known to car $i$, and $\tilde{\eta}$ represents a bounded unknown uncertainty.

\subsection{Considered Attacker Capabilities}\label{ssec:attack_capabilities}
\noindent The wireless connectivity in the CACC platoon aids performance, but also makes the cars vulnerable to cyber-attacks. In this work, an attacker is considered which can perform \emph{man-in-the-middle} attacks on both the communicated input $u_{i-1}$, and all measurements $\tilde{y}_i$. The aim is to determine what attack can do most harm to the system without being detected.

The attacks are denoted as data injection attacks, such that
\begin{equation*}
    u_{i-1,a} = u_{i-1} + \Delta{u}_{i-1}\,; \tilde{y}_{i,a} = \tilde{y}_i + \Delta \tilde{y}\,,
\end{equation*}
where $\Delta{u}_{i-1}$ and $\Delta \tilde{y}$ are the bounded cyber-attacks, and $u_{i-1,a}$ and $\tilde{y}_{i,a}$ are the attacked signals known to car $i$.

The model can then be rewritten as
\begin{equation}\label{eq:full_CACC_model}
    \left\{
\begin{aligned}
\dot{\tilde{x}} &= \hat{A}\tilde{x}+\hat{B}u_a+E\tilde{\eta}+F\Delta u\,,\\
 \tilde{y}_{i,a} &= C\tilde{x}+\tilde{c}+\Delta \tilde{y}+\tilde{\zeta}_i\,,
 \end{aligned}\right.
\end{equation}
\begin{equation*}
F=\left[\begin{matrix}
0_{2\times 1}\\
-\frac{1}{\hat{\tau}_{i-1}}\\
0_{3\times 1}
\end{matrix}\right]~;~u_a=\left[\begin{matrix}u_{i-1,a}\\u_i\end{matrix}\right]\,,
\end{equation*}
such that $u_a$ and $F$ are known to car $i$.

The considered attack classes are now introduced.
\begin{definition}[\cite{pasqualetti2013attack}]\label{def:stealthy}
\noindent
Consider a non-zero cyber-attack $\Delta$. This cyber-attack is said to be \emph{stealthy}
if and only if $\tilde{y}_{i,a}(\Delta,t)=\tilde{y}_{i,a}(0,t)$ for all $t\in\mathbb{R}_{\geq0}$.
\end{definition}

\begin{definition}[Extended from \cite{Ao2016}]\label{def:quantifiable}
An attack $\Delta$ is said to be \emph{quantifiable} if there exists a smooth function $r_\Delta (\tilde{y}_{i,a},u_{a},t)$ such that $\lim_{t\to\infty} r_\Delta(\tilde{y}_{i,a},u_{a},t) \in \Delta +\left[-\delta, \delta \right]$ for a bounded $\delta$.
\end{definition}

\begin{definition}\label{def:nonstealthy}
An attack is said to be \emph{non-stealthy} if it is not stealthy and not quantifiable.
\end{definition}

\begin{remark}
In \cite{pasqualetti2013attack} a linear observer is proposed, and it is proven that the attacked and healthy residual are indistinguishable under stealthy attacks.
In \cite{Ao2016} a SMO-based residual is used to prove that an attack is quantifiable if the output and state equations are not attacked simultaneously. \end{remark}

\subsection{Extended System for Observer Implementation}\label{ssec:extended_sys}
\noindent In section \ref{sec:smo}, an SMO-based attack detection approach will be presented. The used SMO based approach originally only considers attacks to appear in the state equation \cite{Keijzer2019,Edwards2000-vs}. In \cite{Tan2003b}, also attacks on part of the output equation are considered. To this end, the system is transformed and extended such that all attacks again appear in the state equation only. In this work the extended system from \cite{Tan2003b} is used with the observer and detection threshold from \cite{Keijzer2019}. However, here an even broader class of attacks is considered, where all outputs are subject to attacks. To accommodate this, the output attack will be split up such that one part is added to the state equation, and one part will remain in the output equation.

Extending the approach in \cite{Tan2003b}, the output of system \eqref{eq:full_CACC_model} is transformed using $\left[\begin{matrix}y_{i,1}\\y_{i,2}\end{matrix}\right]=T_y \tilde{y}_{i,a}$, such that $\left[\begin{matrix}\Delta y_{1}\\ \Delta y_{2}\end{matrix}\right]=T_y \Delta \tilde{y}$
and $y_{i,2}\in\mathbb{R}^{h}$ can be added to the state equation. Here $T_y$ is a permutation matrix. With this
\begin{equation*}
    \left\{
\begin{aligned}
\dot{\tilde{x}} &= \hat{A}\tilde{x}+\hat{B}u_a+E\tilde{\eta}+F\Delta u\,,\\
 y_{i,1} &= \mathcal{C}_1\tilde{x}+c_1+\Delta y_1+\zeta_{i,1}\,,\\
 y_{i,2} &= \mathcal{C}_2\tilde{x}+c_2+\Delta y_2+\zeta_{i,2}\,,
 \end{aligned}\right.
\end{equation*}
where$\left[\begin{matrix}\mathcal{C}_1\\\mathcal{C}_2\end{matrix}\right]=T_yC$, $\left[\begin{matrix}c_1\\c_2\end{matrix}\right]=T_y\tilde{c}$, $\left[\begin{matrix}\zeta_{i,1}\\\zeta_{i,2}\end{matrix}\right]=T_y\tilde{\zeta}_i$ and $\Delta y = \left[\begin{matrix}\Delta y_1\\ \Delta y_2\end{matrix}\right]$. Note that $T_y$ and $h$ (dimension of $y_{i,2}$) are design choices that determine which elements of the output are added to the state equation. The choices for $T_y$ and $h$ are constrained by assumptions \ref{ass:polesE} and \ref{ass:polesF},which are used in proves later in the paper. The final design choices are presented in section \ref{sec:design}.

Next, the part of the output $y_{i,2}$ is filtered using $\dot{z}=A_f(z-y_{i,2})$,
such that it can be added to the state equation. Here $A_f\prec 0$ is the filter gain matrix. This allows to write
\begin{equation*}
\left\{
    \begin{aligned}
    \left[\begin{matrix}\dot{\tilde{x}}\\\dot{z}\end{matrix}\right]=&
    \underset{A_e}{\underbrace{\left[\begin{matrix}\hat{A}&0\\-A_f\mathcal{C}_2&A_f\end{matrix}\right]}}
    \underset{x}{\underbrace{\left[\begin{matrix}\tilde{x}\\z\end{matrix}\right]}}
    +\underset{B_e}{\underbrace{\left[\begin{matrix}\hat{B}&0\\0&-A_f\end{matrix}\right]}}
    \underset{u}{\underbrace{\left[\begin{matrix}u_a\\c_2\end{matrix}\right]}}\\
    &\hspace{0.8cm}+\underset{E_e}{\underbrace{\left[\begin{matrix}E&0\\0&-A_f\end{matrix}\right]}}
    \underset{\eta}{\underbrace{\left[\begin{matrix}\tilde{\eta}\\\zeta_{i,2}\end{matrix}\right]}}
    +\underset{F_e}{\underbrace{\left[\begin{matrix}F&0\\0&-A_f\end{matrix}\right]}}
    \underset{\Delta}{\underbrace{\left[\begin{matrix}\Delta u\\\Delta y_2\end{matrix}\right]}}\\
    \underset{y}{\underbrace{\left[\begin{matrix}y_{i,1}\\z\end{matrix}\right]}}=&
    \underset{C_e}{\underbrace{\left[\begin{matrix}\mathcal{C}_1&0\\0&I_h\end{matrix}\right]}}
    \left[\begin{matrix}\tilde{x}\\z\end{matrix}\right]
    +\underset{c}{\underbrace{\left[\begin{matrix}c_1\\0\end{matrix}\right]}}+\underset{D}{\underbrace{\left[\begin{matrix}I_{p-h}\\0\end{matrix}\right]}}\zeta_{i,1}+\underset{H}{\underbrace{\left[\begin{matrix}\mathcal{H}\\0\end{matrix}\right]}}\Delta y_1
    \end{aligned}\right.
\end{equation*}
In this extended system, part of the attack on the output, $\Delta y_2$, appears in the state equation. Therefore, if $\Delta y_1=0$, the attack can be considered within the original framework of the SMO based attack detection approach \cite{Keijzer2019,Edwards2000-vs,Tan2003b}.

The state of the extended system will then be transformed as $\left[\begin{matrix}x_1\\x_2\end{matrix}\right]=Tx$, such that $C_e T^{-1}=\left[\begin{matrix}0\\I_{p}\end{matrix}\right]^\top$. This results in
\begin{equation}\label{eq:SMO_system}
\left\{
    \begin{aligned}
    \left[\begin{matrix}\dot{x}_1\\\dot{x}_2\end{matrix}\right]=&
    \left[\begin{matrix}A_{11}&A_{12}\\A_{21}&A_{22}\end{matrix}\right]\left[\begin{matrix}x_1\\x_2\end{matrix}\right]
    +\left[\begin{matrix}B_1\\B_2\end{matrix}\right]u+\left[\begin{matrix}E_1\\E_2\end{matrix}\right]\eta
    +\left[\begin{matrix}F_1\\F_2\end{matrix}\right]\Delta\\
    y=&x_2
    +c+D\zeta_{i,1}+H\Delta y_1
    \end{aligned}\right.
\end{equation}
where $y\in\mathbb{R}^p$, $x_1\in\mathbb{R}^{n-p}$, $x_2\in\mathbb{R}^p$, $u\in\mathbb{R}^{2+h}$, $\eta\in\mathbb{R}^{1+h}$, $\Delta\in\mathbb{R}^{1+h}$, $\zeta_{i,1}\in\mathbb{R}^{p-h}$, and $\Delta y_{1}\in\mathbb{R}^{p-h}$. In this form the system can be used for the implementation of the SMO-based attack detection approach. Throughout this section assumptions were mentioned, which are summarized below in terms of the variables of system \eqref{eq:SMO_system}.
\begin{assumption}\label{ass:bound}
The cyber-attack $\Delta$, model uncertainty $\eta$, and measurement noise $\zeta_{i,1}$ are bounded as $\lvert\Delta\rvert\leq\bar{\Delta}$, $\lvert\eta\rvert\leq\bar{\eta}$, and $\lvert\zeta_{i,1}\rvert\leq\bar{\zeta}_{i,1}$ by known $\bar{\Delta}$, $\bar{\eta}$, and $\bar{\zeta}_{i,1}$, respectively.
\end{assumption}
\section{Detection \& Estimation Method}\label{sec:smo}
\noindent In this section, the SMO-based cyber-attack estimation and detection method from \cite{Keijzer2019} will be presented. Furthermore, the resulting error dynamics are presented. Lastly, the \emph{Equivalent Output Injection} (EOI) is introduced, which will form the basis of the attack estimation presented in section \ref{sec:est}.

\subsection{Sliding Mode Observer}
The dynamics of the SMO from \cite{Keijzer2019} can be written as
\begin{equation}\label{eq:SMO}
\left\{
    \begin{aligned}
    \left[\begin{matrix}\dot{\hat{x}}_1\\\dot{\hat{x}}_2\end{matrix}\right]=&
    \left[\begin{matrix}A_{11}&A_{12}\\A_{21}&A_{22}\end{matrix}\right]\left[\begin{matrix}\hat{x}_1\\\hat{x}_2\end{matrix}\right]
    +\left[\begin{matrix}B_1\\B_2\end{matrix}\right]u-\left[\begin{matrix}A_{12}\\A_{22}^{-s}\end{matrix}\right]e_y
    +\left[\begin{matrix}0\\\nu\end{matrix}\right]\,,\\
    \hat{y}=&\hat{x}_2+c\,,\\
    \nu =& -\rho \text{sgn}(e_y)\,,
    \end{aligned}\right.
\end{equation}
where $e_y\triangleq\hat{y}-y$, $A_{22}^{-s}=A_{22}-A_{22}^s$, and $\rho\succ 0$ and $A_{22}^s\prec 0$ are diagonal matrices.
When applied to system \eqref{eq:SMO_system}, this results in the error dynamics
\begin{equation}\label{eq:SMO_err_dyn}
    \left\{
    \begin{aligned}
    \left[\begin{matrix}\dot{e}_1\\\dot{e}_2\end{matrix}\right]=&
    \left[\begin{matrix}A_{11}&0\\A_{21}&A_{22}^s\end{matrix}\right]\left[\begin{matrix}e_1\\e_2\end{matrix}\right]-\left[\begin{matrix}E_1\\E_2\end{matrix}\right]\eta\\
    &\hspace{0.3cm}+\left[\begin{matrix}A_{12}\\A_{22}^{-s}\end{matrix}\right](D\zeta_{i,1}+H\Delta y_1)
    -\left[\begin{matrix}F_1\\F_2\end{matrix}\right]\Delta
    +\left[\begin{matrix}0\\\nu\end{matrix}\right]\,,\\
    e_y=&e_2-D\zeta_{i,1}-H\Delta y_1\,,\\
    \end{aligned}\right.
\end{equation}
where $e_1=\hat{x}_1-x_1$ and $e_2=\hat{x}_2-x_2$.

We will now recall a result from \cite{Keijzer2019}, which holds under the following assumption 
\begin{assumption}\label{ass:polesE}
All poles of the pair $(A_{11},E_1)\in \mathbb{C}^-$.
\end{assumption}
\begin{lemma}[\cite{Keijzer2019}]\label{prop:err_bounds}
Consider the error dynamics in system \eqref{eq:SMO_err_dyn}. Define $\underline{e}_1\leq e_1 \leq \bar{e}_1$, $\tilde{e}_1\geq \lvert e_1\rvert$, $\bar{e}_2^0\geq\lvert e_2 \rvert$, and $\underline{\dot{e}}_2\leq\lvert \dot{e}_2 \rvert\leq\bar{\dot{e}}_2$. Furthermore denote the healthy condition, when $\Delta=\Delta y_1=0$, with superscript $^0$. 

If Assumptions \ref{ass:bound} and \ref{ass:polesE} hold and, element-wise, $\text{diag}(\rho)>\lvert A_{21}\rvert \tilde{e}_1+\lvert A_{22}\rvert\bar{\zeta}_{i,1}+\lvert E_2\rvert\bar{\eta}+\lvert F_2\rvert\bar{\Delta}$, then
\begin{equation*}
    \begin{aligned}
    \bar{e}_1^0&=e^{A_{11}(t)}e_1(0) -A_{11}^{-1} (I-e^{A_{11}(t)})(\lvert A_{12}\rvert\bar{\zeta}_{i,1}+\lvert E_1\rvert\bar{\eta})\\
    \underline{e}_1^0&=e^{A_{11}(t)}e_1(0) +A_{11}^{-1} (I-e^{A_{11}(t)})(\lvert A_{12}\rvert\bar{\zeta}_{i,1}+\lvert E_1\rvert\bar{\eta})\\
    \bar{e}_2^0&=\bar{\zeta}_{i,1}\\
    \bar{\dot{e}}^0_2&= \lvert A_{21}\rvert \bar{e}^0_1+(\lvert A_{22}^{-s}\rvert+\lvert A_{22}^s \rvert)\bar{\zeta}_{i,1}+\lvert E_2\rvert\bar{\eta}+ \rho\\
    \underline{\dot{e}}^0_2&=\lvert A_{21}\rvert \underline{e}^0_1-(\lvert A_{22}^{-s}\rvert+\lvert A_{22}^s \rvert)\bar{\zeta}_{i,1}-\lvert E_2\rvert\bar{\eta}+ \rho\\
    &\hspace{-0.3cm}\text{sgn}(\dot{e}_2)=-\text{sgn}(e_y)\hspace{4.5cm}\blacksquare
    \end{aligned}
\end{equation*} 
\end{lemma}
\subsection{Equivalent Output Injection}
\noindent The EOI, which will be used for cyber-attack estimation and detection is defined as
\begin{equation}\label{eq:EOI}
    \dot{\nu}_{\text{fil}} = A_\nu (\nu_{\text{fil}}-\nu)\,,
\end{equation}
where $A_\nu\prec0$ is the filter gain matrix.

\subsection{Detection Threshold}
\noindent The threshold from \cite{Keijzer2019} is based on the combination of two types of behaviour, which together allow the EOI to attain its worst-case healthy value. In general, the worst-case behaviour is obtained if the duration that $\nu>0$ is maximal compared to the duration that $\nu<0$. This behaviour has to be attained while adhering to the bounds on the observer error dynamics as presented in Proposition \ref{prop:err_bounds}.

A detailed design of the threshold can be found in \cite{Keijzer2019}, while the resulting threshold will be presented below.
\begin{theorem}[\cite{Keijzer2019}]\label{thm:threshold}

\noindent Considering system \eqref{eq:SMO_system}, the observer in equation \eqref{eq:SMO} with $\rho$ as in Proposition \ref{prop:err_bounds}, and the EOI as defined in equation \eqref{eq:EOI}. If Assumptions \ref{ass:bound} and \ref{ass:polesE} hold, $\bar{\nu}_\text{fil}$, as defined in equation \eqref{eq:threshold}, bounds the healthy behaviour of the EOI. $\bar{\nu}_\text{fil}$ can therefore act as a cyber-attack detection threshold which guarantees to have no false detection.

Define the sequence of instants $\left\{t_{k}\right\}$ as the times at which $\nu$ changes sign. Furthermore, assume that $\nu>0$ during each period $\left[t_{2k}~~t_{2k+1}\right]$ and define $t_{k-}=t_{2k}-t_{2k-1}$. Then
\begin{equation}\label{eq:threshold}
    \bar{\nu}_\text{fil}(t_{2k}) = e^{A_{\nu}\bar{t}}\bar{\nu}_{\text{fil},0}(t_{2k})+(1-e^{A_{\nu}\bar{t}})\rho
\end{equation}
where
\begin{equation*}
    \begin{aligned}
    \bar{\nu}_{\text{fil},0}(t_{2k}) =&e^{A_{\nu}\tilde{t}_k}\bar{\nu}_{\text{fil}}(t_{2k-2})+(1-2e^{A_{\nu}t_{k+}}+e^{A_{\nu}\tilde{t}_k})\rho \\
    \bar{t} =&\frac{2\bar{e}_2^0}{\bar{\dot{e}}^0_2} ~~;~~
    t_{k+} = \frac{\bar{\dot{e}}^0_2}{\underline{\dot{e}}^0_2}t_{k-} ~~;~~ \tilde{t}_k=t_{k-}+t_{k+}\\
    \end{aligned}
\end{equation*}
The obtained threshold is valid for the period $\left[t_{2k}~~t_{2k+2}\right]$.\hfill $\blacksquare$
\end{theorem}
\section{EOI based Attack Estimation}\label{sec:est}
\noindent In this section two proofs are presented to show the EOI's ability to estimate the cyber-attacks. It will be shown that not all considered cyber-attacks are quantifiable by the EOI. However, a set of sufficient conditions is proposed for which the cyber-attacks are quantifiable. This subset of attacks can be treated using the existing framework in \cite{Keijzer2019,Tan2003b}.

\begin{lemma}\label{lem:attack_est}
Consider a noiseless version of system \eqref{eq:SMO_system}, where $\zeta_{i,1}=0$, the observer in equation \eqref{eq:SMO} where $\text{diag}(\rho)>\lvert A_{21}\rvert \tilde{e}_1+\lvert A_{22}\rvert\bar{\zeta}_{i,1}+\lvert E_2\rvert\bar{\eta}+\lvert F_2\rvert\bar{\Delta}+\lvert H\rvert \Delta \bar{\dot{y}}_1 + \lvert A_{22} H\rvert \Delta \bar{y}_1$, and the EOI as defined in equation \eqref{eq:EOI}. Furthermore, define $\Delta \bar{y}_1\geq \lvert\Delta y_1\rvert$ and $\Delta \bar{\dot{y}}_1\geq \lvert\Delta \dot{y}_1\rvert$. If Assumption \ref{ass:bound} holds, $$\lim_{t\to\infty}\nu_{\text{fil}}\in -A_{21}e_1-F_2\Delta+H\Delta \dot{y}_1-A_{22}H\Delta y_1+\lvert E_2\rvert \left[-\bar{\eta},\bar{\eta}\right]$$
\end{lemma}
\begin{proof}
Consider the candidate Lyapunov function $V=\frac{1}{2} e_y^\top e_y$. Using the lower bound on $\rho$ provided, and $\dot{e}_y=\dot{e}_2-H\Delta \dot{y}_1$ (eq. \eqref{eq:SMO_err_dyn}) it can be proven that if Assumptions \ref{ass:bound} holds,
\begin{equation*}
    \begin{aligned}
        \dot{V} =& e_y^\top (A_{21} e_1 + A_{22}^s e_y  -E_2 \eta - F_2 \Delta -H\Delta \dot{y}_1 + \\&A_{22} H\Delta y_1 -\rho \text{sgn}(e_y))\\
        \dot{V} < &e_y^\top A_{22}^s e_y\,.
    \end{aligned}
\end{equation*}
This proves that $\lim_{t\to\infty} e_y= 0$. Therefore, from equation \eqref{eq:SMO_err_dyn}
$$\lim_{t\to\infty}\nu\in -A_{21}e_1-F_2\Delta+H\Delta \dot{y}_1-A_{22}H\Delta y_1+\lvert E_2\rvert \left[-\bar{\eta},\bar{\eta}\right].$$
Furthermore, as $\nu_{\text{fil}}\to \nu$ asymptotically, also 
$$\lim_{t\to\infty}\nu_{\text{fil}} \in -A_{21}e_1-F_2\Delta+H\Delta \dot{y}_1-A_{22}H\Delta y_1+\lvert E_2\rvert \left[-\bar{\eta},\bar{\eta}\right].$$
This proves the lemma.
\end{proof}
From Lemma \ref{lem:attack_est} it is clear that the EOI is in general affected by the cyber-attack, and will become non-zero if the attack is not carefully designed by the malicious agent. Furthermore, it can be seen that the EOI is affected by both $\Delta$ and $\Delta y_1$, such that in general they cannot be separately estimated.
\subsection{Sufficient Conditions for Quantifiable Attacks}
\noindent Below, additional assumptions are presented under which it will be proven that the attacks are quantifiable. These attacks can be treated using the existing framework in \cite{Keijzer2019,Tan2003b}.
\begin{assumption}\label{ass:noH}
There are no cyber-attacks affecting directly the output, i.e. $\Delta y_1=0$
\end{assumption}
\begin{assumption}\label{ass:polesF}
All poles of the pair $(A_{11},F_1)\in \mathbb{C}^-$
\end{assumption}
\begin{assumption}\label{ass:rank}
$\text{rank}(A_{21}A_{11}^{\dagger}F_1-F_2)=\text{rank}(F_e)=1+h$
\end{assumption}
\begin{lemma}\label{lem:equivalent}
Assumption \ref{ass:rank} holds iff that $p\geq 1+h$, i.e.
the number of cyber-attacks introduced in the state equation is at most equal to the number of outputs.
\end{lemma}
\begin{proof}
\textbf{if:} $(A_{21}A_{11}^{\dagger}F_1-F_2)$ is a $p\times (1+h)$ matrix, which means its rank is at most $\min(p,1+h)$. Therefore, if Assumption \ref{ass:rank} holds, $p\geq 1+h$. \textbf{only if:} If $p<1+h$, then the rank of $(A_{21}A_{11}^{\dagger}F_1-F_2)$ is at most $p$ so Assumption \ref{ass:rank} cannot hold.
\end{proof}
\begin{remark}
From lemma \ref{lem:equivalent} it can be concluded that not all considered attacks can be added to the state equation. Doing this would result in $h=p$, for which $p \ngeq 1+h$.
\end{remark}
\begin{theorem}\label{thm:attack_est}
Consider a noiseless system \eqref{eq:SMO_system}, where $\zeta_{i,1}=0$, the observer in equation \eqref{eq:SMO} with $\rho$ as in proposition \ref{prop:err_bounds}, and the EOI as in equation \eqref{eq:EOI}. If Assumptions \ref{ass:bound}-\ref{ass:rank} hold,
\begin{equation*}
    \begin{aligned}
    \lim_{t\to\infty}&(A_{21}A_{11}^\dagger F_1-F_2)^\dagger\nu_{\text{fil}}- \Delta\in\\ &(A_{21}A_{11}^\dagger F_1-F_2)^\dagger( A_{21}A_{11}^\dagger \lvert E_1\rvert+\lvert E_2\rvert) \left[ -\bar{\eta},~\bar{\eta}\right]
    \end{aligned}
\end{equation*}
which in turn implies the attack is quantifiable.
\end{theorem}
\begin{proof}
From Lemma \ref{lem:attack_est} it can be inferred that if assumptions \ref{ass:bound} and \ref{ass:noH} hold, $\nu_{\text{fil}}\to -A_{21}e_1-F_2\Delta+\lvert E_2\rvert \left[-\bar{\eta},\bar{\eta}\right]$ asymptotically. Then, consider the dynamics of $e_1$ in equation \eqref{eq:SMO_err_dyn}. If and only if Assumptions \ref{ass:polesE} and \ref{ass:polesF} hold, the Final Value theorem can be applied to prove that $$\lim_{t\to\infty}e_1\in-A_{11}^\dagger (\lvert E_1 \rvert \left[ -\bar{\eta},~\bar{\eta}\right]+ F_1\Delta)\,.$$ Therefore, $$\lim_{t\to\infty}\nu_{\text{fil}}\in (A_{21}A_{11}^\dagger F_1-F_2)\Delta+(A_{21}A_{11}^\dagger \lvert E_1 \rvert+\lvert E_2\rvert) \left[ -\bar{\eta},~\bar{\eta}\right]\,.$$ Furthermore, if and only if Assumption \ref{ass:rank} holds
\begin{equation*}
    \begin{aligned}
    \lim_{t\to\infty}&(A_{21}A_{11}^\dagger F_1-F_2)^\dagger\nu_{\text{fil}}- \Delta\in\\ &(A_{21}A_{11}^\dagger F_1-F_2)^\dagger( A_{21}A_{11}^\dagger \lvert E_1\rvert+\lvert E_2\rvert) \left[ -\bar{\eta},~\bar{\eta}\right]
    \end{aligned}
\end{equation*}
The proof that this implies the attack is quantifiable can be derived directly from definition \ref{def:quantifiable}.
\end{proof}
\section{Detector Design}\label{sec:design}
\noindent The SMO-based detection scheme presented in Section \ref{sec:smo} has estimation and detection properties which are dependent on the assumptions made. By Lemma \ref{lem:attack_est} and Theorem \ref{thm:threshold}, the observer and threshold exist if Assumptions \ref{ass:bound} and \ref{ass:polesE} hold. In this section the class of observer designs for which these assumptions hold, and thus the observer and threshold exist, are identified. Furthermore, one design is chosen to be used for the analysis in the remainder of the paper. The design is chosen with the aim to make the largest possible class of attacks quantifiable, i.e. such that assumptions \ref{ass:noH}-\ref{ass:rank} hold.

\begin{theorem}\label{thm:obsdesign}
The observer \eqref{eq:SMO} and threshold \eqref{eq:threshold} exist iff $\Delta y_1$ includes the  relative velocity measurement attack $\Delta \tilde{y}_{i,(2)}$.
\end{theorem}
\begin{proof}
For the observer and threshold to exist, assumptions \ref{ass:bound} and \ref{ass:polesE} should hold. Assumption \ref{ass:bound} are not considered here as they are independent of the observer design. Therefore, the observer and threshold exist when $T_y$ and $h$ are chosen such that assumption \ref{ass:polesE} holds.

As there is only a finite number of options for $T_y$ and $h$, Assumption \ref{ass:polesE} can be checked for each possible combination.\footnote{$T_y\in\mathbb{R}^{p\times p}$ is a permutation matrix, giving $p!$ possibilities for $T_y$. There are $p+1$ possibilities for $0\leq h\leq p$, giving $p+1!$ possible combinations.} By doing so, it can be found that for the considered scenario Assumption \ref{ass:polesE} holds iff $T_y$ and $h$ are chosen such that $y_{i,1}$ includes $\tilde{y}_{i,a,(2)}$. This means $\Delta y_1$ must include the attack on the relative velocity measurement $\Delta \tilde{y}_{i,(2)}$.
\end{proof}
The final design is chosen to make most attacks quantifiable. Using the result of lemma \ref{lem:equivalent} and theorem \ref{thm:obsdesign}, we can choose $h=p-1$ such that all attacks except $\Delta \tilde{y}_{i,(2)}$ are added to the state equation giving the design $h=3$, $T_y=\left[\begin{matrix}e_2&e_1&e_3&e_4\end{matrix}\right]$, where $e_i$ are the standard basis vectors.

In conclusion, the detector is designed such that the largest possible class of attacks is quantifiable. This is the case when $\Delta y_1=\Delta \tilde{y}_{i,(2)}$. With the chosen design assumptions \ref{ass:polesE}, \ref{ass:polesF}, and \ref{ass:rank} are inherently satisfied. Furthermore, Assumption \ref{ass:noH}, required for an attack to be quantifiable, is only satisfied when the relative velocity measurement $\tilde{y}_{i,(2)}$ is not attacked.
\section{Cyber-Attack Classification}\label{sec:attack}
\noindent In this section, the considered attacks will be classified according to definitions \ref{def:stealthy}-\ref{def:nonstealthy}, for the detector design presented in section \ref{sec:design}.

In theorem \ref{thm:attack_est} it is proven that attacks are detectable if assumption \ref{ass:noH} holds. For the chosen design this means that the relative velocity measurement $\tilde{y}_{i,(2)}$ is not attacked. Therefore, all attacks for which $\Delta \tilde{y}_{i,(2)} = 0$ are quantifiable.

Below, in lemma \ref{lem:undetectable} and theorem \ref{thm:undetectable} it will be shown that only attacks of a specific form are stealthy, and an analytical expression for these stealthy attacks is presented. 
\begin{lemma}\label{lem:undetectable}
Consider system \eqref{eq:SMO_system}. The following statements are equivalent:
\begin{enumerate}
    \item An attack is stealthy when $\zeta_{i,1}=0$
    \item An attack is stealthy when $\zeta_{i,1}\neq0$
\end{enumerate}
\end{lemma}
\begin{proof}
An attack is stealthy if $y(\Delta,\Delta y_1,t)=y(0,0,t)$. 
\begin{enumerate}
    \item When $\zeta_{i,1}= 0$, $y(\Delta,\Delta y_1,t)=y(0,0,t)$ is equivalent to $x_2(\Delta)+c+H\Delta y_1=x_2(0)+c\leftrightarrow x_2(\Delta)+H\Delta y_1=x_2(0)$.
    \item When $\zeta_{i,1}\neq 0$ this is equivalent to $x_2(\Delta)+c+D\zeta_{i,1}+H\Delta y_1=x_2(0)+c+D\zeta_{i,1} \leftrightarrow x_2(\Delta)+H\Delta y_1=x_2(0)$.
\end{enumerate} 
The final conditions resulting from both statements are the same, therefore the statements are equivalent.
\end{proof}

\begin{theorem}\label{thm:undetectable}
Consider the system \eqref{eq:SMO_system} with the observer in equation \eqref{eq:SMO} using the design from section \ref{sec:design} where $y_{i,1}=\tilde{y}_{i,a,(2)}$. Furthermore, consider $\nu_\text{fil}$ as defined in equation \eqref{eq:EOI}.
Then, the following statements are equivalent
\begin{enumerate}
    \item An attack is stealthy
    \item $\nu_\text{fil}(\Delta,\Delta y_1,t)=\nu_\text{fil}(0,0,t)~ \forall t\in\mathbb{R}_{\geq0}$
    \item The attack is designed as \begin{equation*}
\left\{\begin{aligned}
    \Delta_{(1)}&=-\hat{\tau}_{i-1}\Delta \ddot{y}_1-\Delta \dot{y}_1\\
    \dot{\Delta}_{(2)} &= -\Delta y_1\\
    \Delta_{(3)} &= \Delta_{(4)} = 0
    \end{aligned}\right.
\end{equation*}
\end{enumerate}
\end{theorem}
\begin{proof}
Using the result from lemma \ref{lem:undetectable}, without loss of generality, we can consider the system where $\zeta_{i,1}=0$. First, prove the equivalence between statements 1) and 3).

A stealthy attack is defined as an attack for which $y(x_2(\Delta),\Delta y_1)=y(x_2(0),0)$. Using system \eqref{eq:SMO_system} this condition can be rewritten as an explicit condition on the cyber-attack. 
\begin{equation*}
    \begin{aligned}
    x_2(\Delta)+c+H\Delta y_1 = x_2(0)+c\\
    x_2(\Delta)-x_2(0) = -H\Delta y_1
    \end{aligned}
\end{equation*}
By taking the double time derivative of this relation we get
\begin{equation*}
    \begin{aligned}
    A_{21}A_{11}(x_1(\Delta)-x_1(0)) =&\\
    -H\Delta \ddot{y}_1+A_{22}H\Delta \dot{y}_1 +& A_{21}A_{12}H\Delta y_1- F_2\dot{\Delta} - A_{21}F_1 \Delta
    \end{aligned}
\end{equation*}
This relation can be written explicitly for the system \eqref{eq:SMO_system} as
\begin{equation*}
    \left\{\hspace{-0.1cm}\begin{aligned}
    \Delta \ddot{y}_1&=-\frac{1}{\hat{\tau}_{i-1}}(\Delta \dot{y}_1+\Delta_{(1)})+\frac{1}{\tau_{i}}(x_{1,(5)}(\Delta)-x_{1,(5)}(0))\\
    \dot{\Delta}_{(2)} &=\Delta y_1\\
    \dot{\Delta}_{(3)} &= x_{1,(5)}(\Delta)-x_{1,(5)}(0)\\
    \dot{\Delta}_{(4)} &= -\frac{1}{\tau_{i}}(x_{1,(5)}(\Delta)-x_{1,(5)}(0))\\
    \end{aligned}\right.
\end{equation*}
Furthermore, the dynamics of $x_1(\Delta)-x_1(0)$ if can be found that $x_{1,(5)}$ is not affected by the attacks, and thus $x_{1,(5)}(\Delta)-x_{1,(5)}(0)=0$, allowing to simplify
\begin{equation*}
    \left\{\begin{aligned}
    \Delta \ddot{y}_1&= -\frac{1}{\hat{\tau}_{i-1}}(\Delta \dot{y}_1+\Delta_{(1)})\\
    \dot{\Delta}_{(2)} &= \Delta y_1\\
    \Delta_{(3)} &= 0 ~;~
    \Delta_{(4)} = 0\\
    \end{aligned}\right.
\end{equation*}
which is equivalent to the system in statement 3). This proves equivalence between statements 1) and 3). To prove equivalence between statements 1) and 2), use that $e_y(0,0)=0$ once the steady state is reached (Lemma \ref{lem:attack_est}) below.
\begin{equation*}
    \begin{aligned}
        0=&\nu_\text{fil}(e_y(\Delta,\Delta y_1))-\nu_\text{fil}(e_y(0,0))\\
        0=&\dot{\nu}_\text{fil}(e_y(\Delta,\Delta y_1))-\dot{\nu}_\text{fil}(e_y(0,0))\\
        0=&A_v\rho (\text{sgn}(e_y(\Delta,\Delta y_1))-\text{sgn}(e_y(0,0)))\\
        0=&\text{sgn}(e_y(\Delta,\Delta y_1))\\
        0=&e_y(\Delta,\Delta y_1)=e_y(0,0)\\
        &\hat{y}(0)-y(x_2(\Delta),\Delta y_1)=\hat{y}(0)-y(x_2(0),0)\\
        0=&y(x_2(\Delta),\Delta y_1)-y(x_2(0),0)
    \end{aligned}
\end{equation*}
This proves the equivalence between statements 1) and 2)
\end{proof}

Lastly, all attacks that are neither quantifiable nor stealthy, are non-stealthy.
\section{Simulation Study}\label{sec:sim}
\noindent In this section simulation results will be presented with the aim to analyse how harmful attacks can be to the system while avoiding detection. To this end, three simulation scenarios will be shown, each with a carefully crafted attack within one of the attack classes. The attacks are crafted using the attacker objectives as described in section \ref{ssec:attack_goal}.

Below, first the considered CACC platoon will be presented, including all controller and observer gains. Then, the different cyber-attacks are introduced, and the corresponding simulation results are analysed.

\subsection{CACC Scenario and Parameters}
\noindent For the presented CACC scenario, the considered platoon is driving at a constant speed of $8.5 \left[\frac{m}{s}\right]\approx 30\left[\frac{km}{h}\right]$. This speed is a common speed limit within urban environments.
All parameters used in simulation are shown in table \ref{tab:param}.
\subsection{Attacker Goal}\label{ssec:attack_goal}
\noindent The cyber-attacks are designed by the attacker with the goal of crashing the cars, i.e. $\tilde{y}_{i,(1)}\leq 0$. To obtain this goal the applied attacks should send car $i$ information that signals that car $i-1$ is further away than desired, or moving away from car $i$. The attack that achieves this is
\begin{equation} \label{eq:attack_goal}
    \begin{aligned}
    \Delta \tilde{y}_{(1)} ,\Delta \tilde{y}_{(2)}, \Delta u& ~~ \text{maximal} \\
    \Delta \tilde{y}_{(3)} ,\Delta \tilde{y}_{(4)}& ~~ \text{minimal}
    \end{aligned}
\end{equation}
\begin{table}[h]
    \centering
    \caption{Parameters used in simulation}
    \label{tab:param}
    \begin{tabular}{c|c||c|c}
        Parameter & Value & Parameter & Value \\\hline
        $p_0(0)$ & $0\,[m]$ & $p_1(0)$&$-11.45\,[m]$\\
        $v_0(0)$ & $8\,\left[\frac{m}{s}\right]$ & $v_1(0)$ & $8\,\left[\frac{m}{s}\right]$\\
        $a_0(0)$ & $0\,\left[\frac{m}{s^2}\right]$& $a_1(0)$&$0\,\left[\frac{m}{s^2}\right]$\\\hline
        $\tau_0$ &$0.1\,[s]$ & $\tau_1$& $0.1\,[s]$\\
        $L_1$ & $4\,[m]$ &$r_{\tau}$& $1.1\,[-]$\\\hline
        $h$ & $0.7\,[s]$ & $r$& $1.5\,[m]$\\
        $k_p$ & $0.2\,[s^{-2}]$ &$k_d$ & $0.7\,[s^{-1}]$\\\hline
        $\bar{\Delta}$ & \multicolumn{3}{c}{$[10~10~10~10]^\top\,\left[\frac{m}{s^2}~m~\frac{m}{s}~\frac{m}{s^2}\right]$}\\
         $\bar{\eta}$&\multicolumn{3}{c}{$[1~0.15~0.03~0.15]^\top\,\left[\frac{m}{s^2}~m~\frac{m}{s}~\frac{m}{s^2}\right]$} \\
        $\bar{\zeta}_1$ & $0.3\frac{m}{s}$&$A_\text{fil}$&$-5\cdot I_3 [s^{-1}]$\\\hline
        $A_\nu$ & $I_4\,[s^{-1}]$& $A_{22}^s$& $- I_4\,[s^{-1}]$\\
        $M$ & \multicolumn{3}{c}{$\text{diag}([11.5~11~11~11])\,\left[m~\frac{m}{s}~\frac{m}{s}~\frac{m}{s^2}\right]$}\\
    \end{tabular}
\end{table}
\subsection{Simulation with quantifiable Attack}
\noindent Simulation results for a quantifiable attack are shown on the left in figures 1-3. All attack signals are chosen to be steps. The size and sign of the attacks are chosen according to the goal in equation \eqref{eq:attack_goal}, while staying undetected.

In figure 1, it can be seen that the attacks are correctly estimated. Furthermore, from figure 2 it can be seen that the attacks affect the residual, but are not detected by the thresholds. Lastly, in figure 3 it can be seen that the attack causes a crash between the cars. However, the relative velocity at the time of the crash is low.

In general, it can be concluded that for quantifiable attacks significant, but limited harm can be done to the system. The maximum attack impact can be limited by reducing the system uncertainty, which causes tighter thresholds.

\subsection{Simulation with Stealthy Attack}
\noindent Simulation results for a stealthy attack are shown on the right in figures 1-3. This attack is designed as presented in theorem \ref{thm:undetectable}, where $\Delta y_1$ is a filtered ramp with positive slope.

In figure 1, it can be seen that none of the attacks are correctly estimated. Furthermore, from figure 2 it can be seen that the attack does not affect the EOI. Therefore, the EOI stays well between the thresholds. Lastly, in figure 3 it can be seen that the stealthy attack causes a high speed crash.

In general, it can be concluded that a stealthy attack can harm the system without bounds. However, perfect model knowledge is required by the attacker to perform this attack.

\subsection{Simulation with non-stealthy Attack}
\noindent Simulation results for a non-stealthy attack are shown in the middle in figures 1-3. The design of a non-stealthy attack is complicated by the non-zero $\Delta y_1$, because  $\nu_{\text{fil},(2)}$ depends on the integral of $\Delta y_1$. Therefore, detection can only be avoided if the integral of $\Delta y_1$ is bounded, meaning a constant attack is not feasible. Instead a sinusoidal attack has been chosen.
It was chosen to set $\Delta_{(1)}$ and $\Delta_{(2)}$ to zero to maximize the potential of the attack on $\Delta y_1$.\footnote{Non-zero choices for $\Delta_{(1)}$ and $\Delta_{(2)}$ would lead to a scenario closer to the quantifiable attack presented in this section.} Furthermore, $\Delta_{(3)}$ and $\Delta_{(4)}$ were chosen the same as for the quantifiable attack.

In figure 1, it can be seen that only attacks $\Delta_{(3)}$ and $\Delta_{(4)}$ are correctly estimated. Furthermore, from figure 2 it can be seen that the attack affects the residual, but is not detected by the threshold. Lastly, in figure 3 it can be seen that this attack has almost no effect on the behaviour of the platoon.

In general, it can be concluded that for non-stealthy attacks the potential harm while avoiding detection is very limited.
\vspace{-0.1cm}
\begin{figure*}
  \begin{minipage}{.33\textwidth}
	\includegraphics[width=1\columnwidth,height=4cm, trim={0cm 0.1cm 7cm 20.5cm},clip]{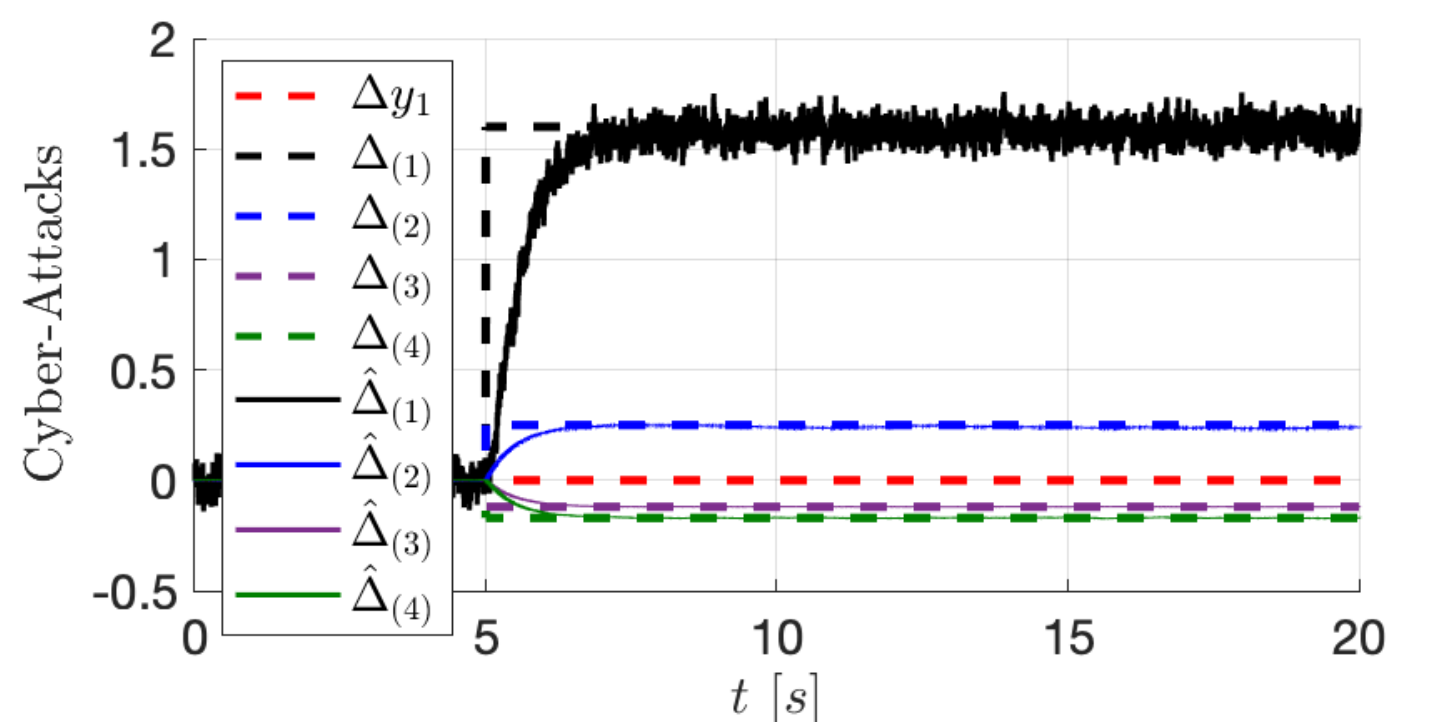}
  \end{minipage} \quad
  \begin{minipage}{.33\textwidth}
    \includegraphics[width=1\columnwidth,height=4cm, trim={1.2cm 0.1cm 7cm 20.5cm},clip]{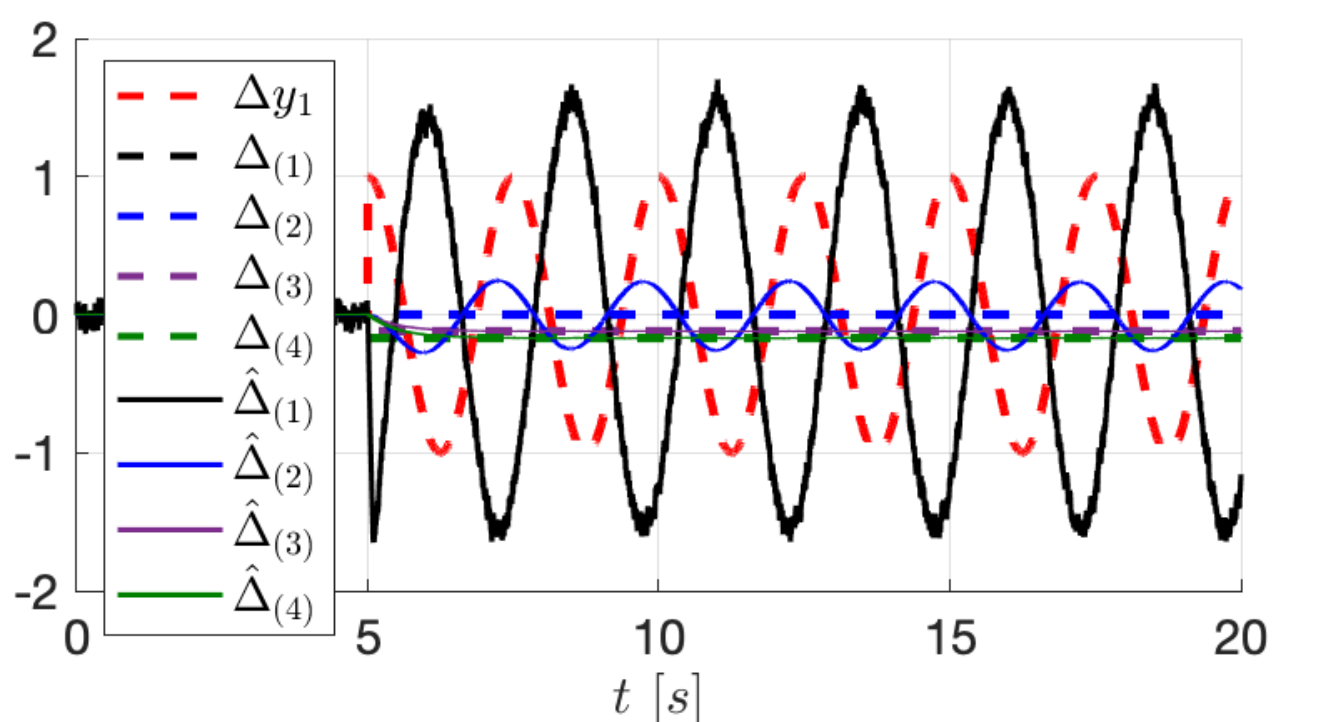}
  \end{minipage}
  \begin{minipage}{.33\textwidth}
   \includegraphics[width=1\columnwidth,height=4cm, trim={1cm 0.1cm 7cm 20.5cm},clip]{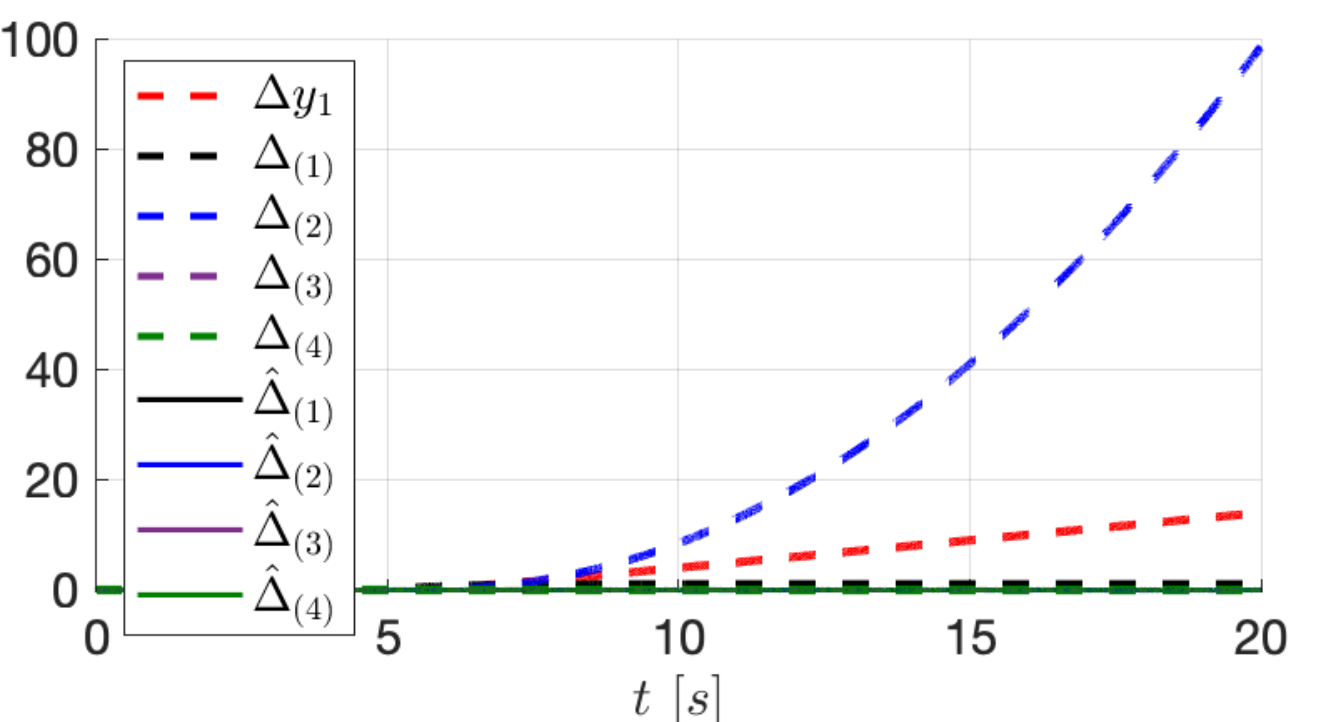}
  \end{minipage}
  \label{fig:estimate}
  \caption{The applied cyber-attacks, and the EOI-based estimates. left: Quantifiable, Middle: non-stealthy, right: Stealthy}
\end{figure*}%
\begin{figure*}
\vspace{-0.35cm}
  \begin{minipage}{.33\textwidth}
	\includegraphics[width=1\columnwidth,height=4cm, trim={0cm 0.1cm 7cm 18.2cm},clip]{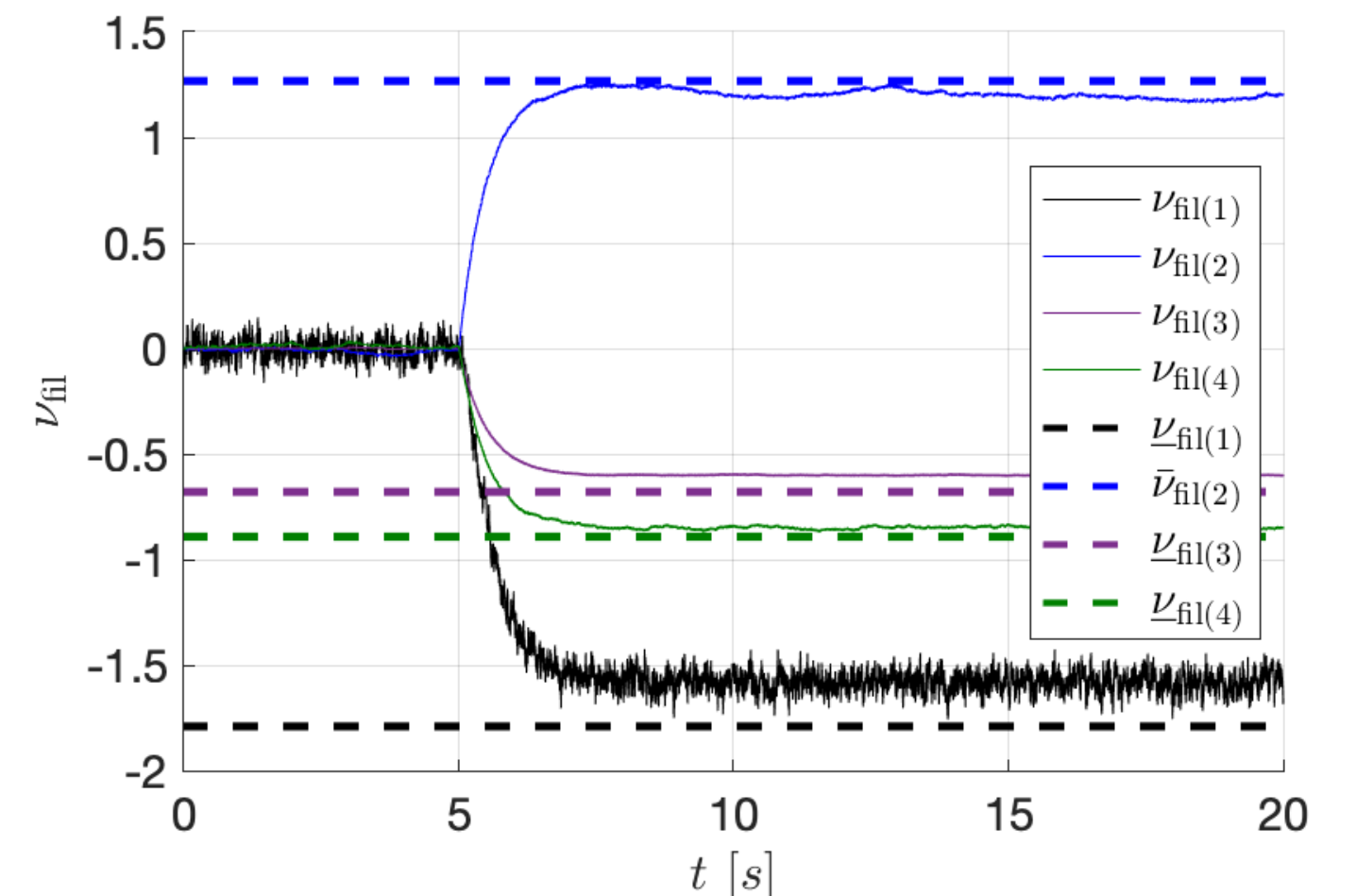}
  \end{minipage} \quad
  \begin{minipage}{.33\textwidth}
    \includegraphics[width=1\columnwidth,height=4cm, trim={1.2cm 0.1cm 7cm 18.2cm},clip]{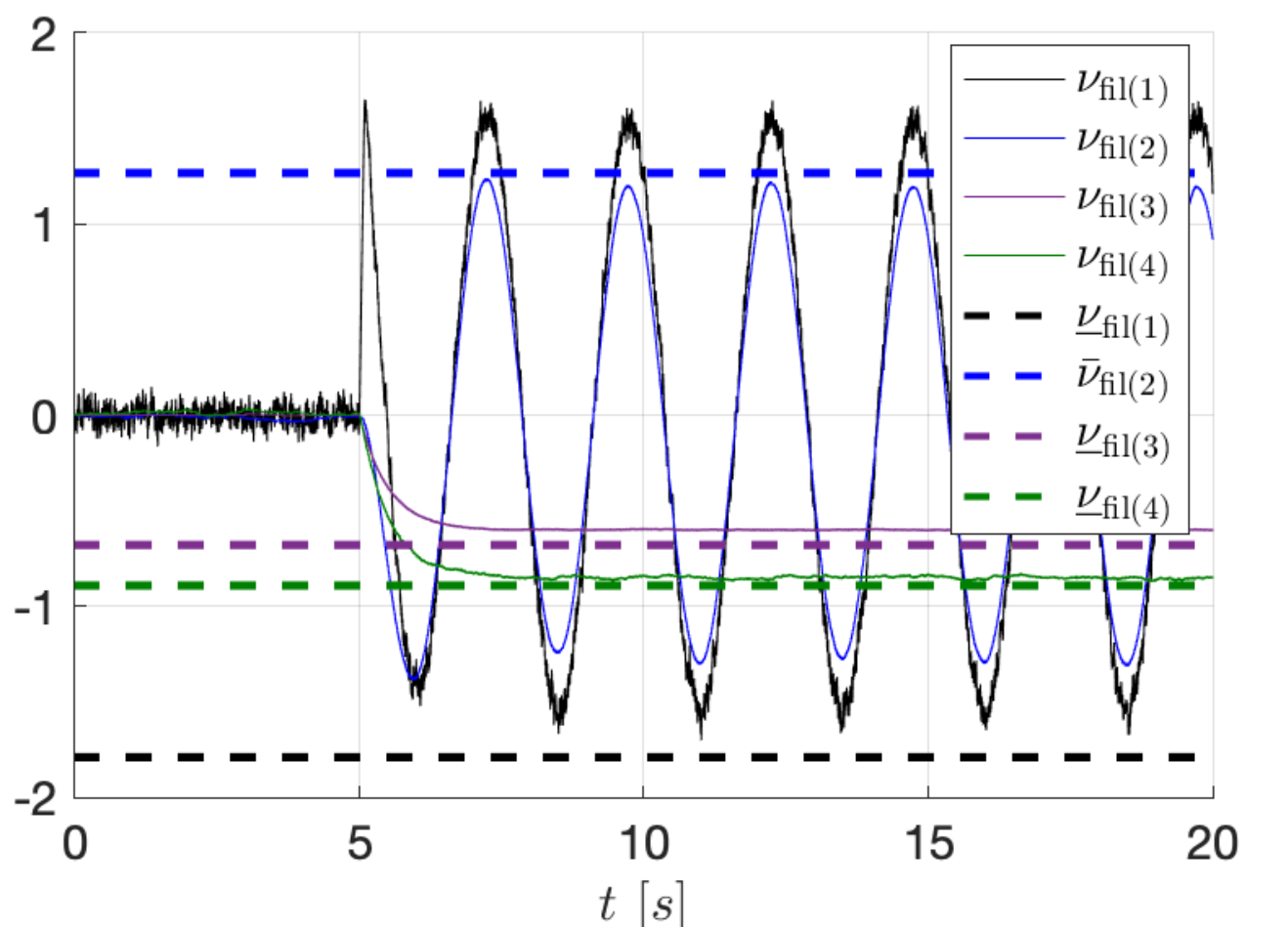}
  \end{minipage}
  \begin{minipage}{.33\textwidth}
   \includegraphics[width=1\columnwidth,height=4cm, trim={1cm 0.1cm 7cm 18.2cm},clip]{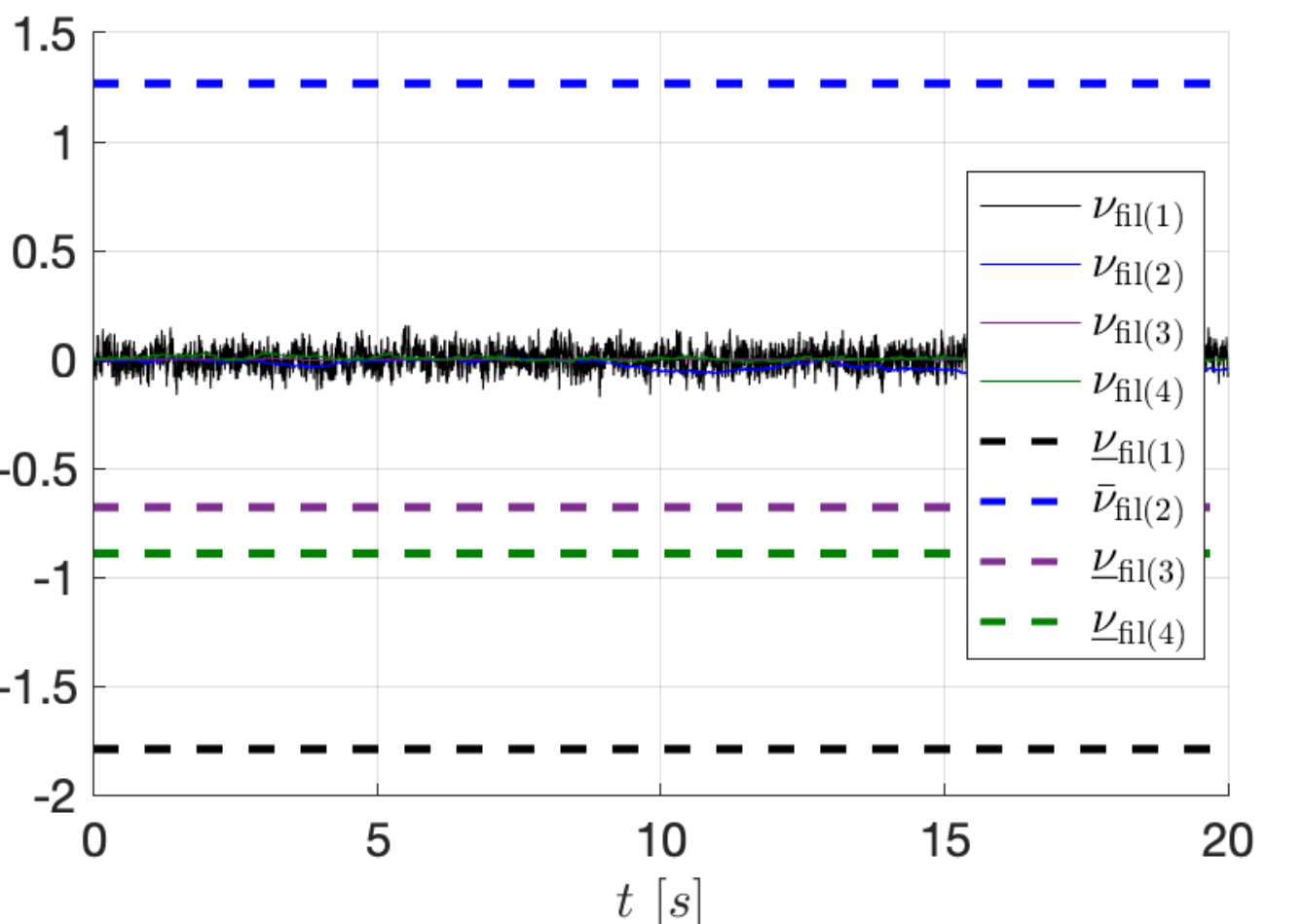}
  \end{minipage}
  \label{fig:eoi}
  \caption{The EOI response to the applied cyber-attacks, and the corresponding thresholds. left: Quantifiable, Middle: non-stealthy, right: Stealthy}
\end{figure*}%
\begin{figure*}
\vspace{-0.35cm}
  \begin{minipage}{.33\textwidth}
	\includegraphics[width=1\columnwidth,height=4cm, trim={0cm 0.1cm 7cm 18.2cm},clip]{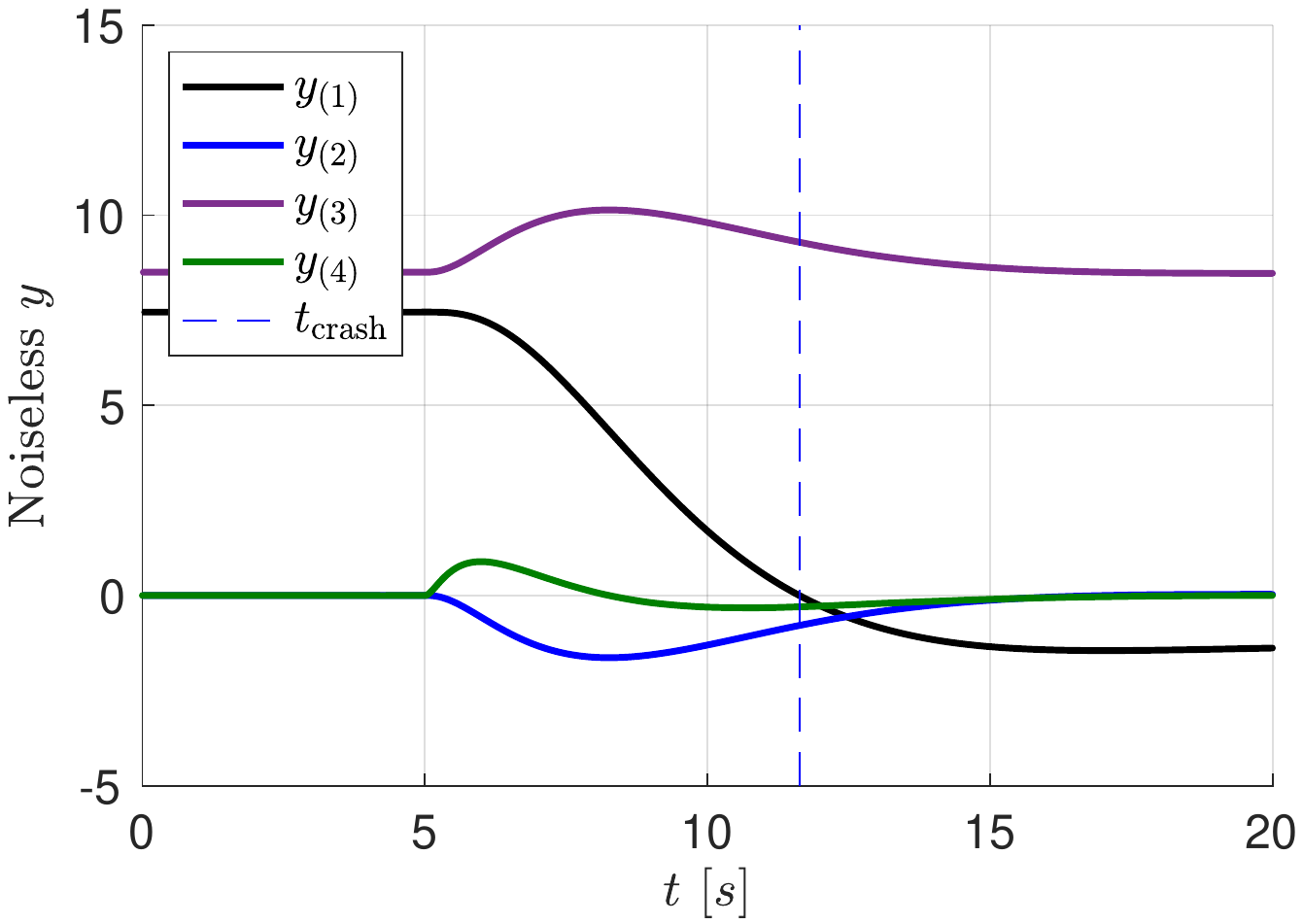}
  \end{minipage} \quad
  \begin{minipage}{.33\textwidth}
    \includegraphics[width=1\columnwidth,height=4cm, trim={1cm 0.1cm 7cm 18.2cm},clip]{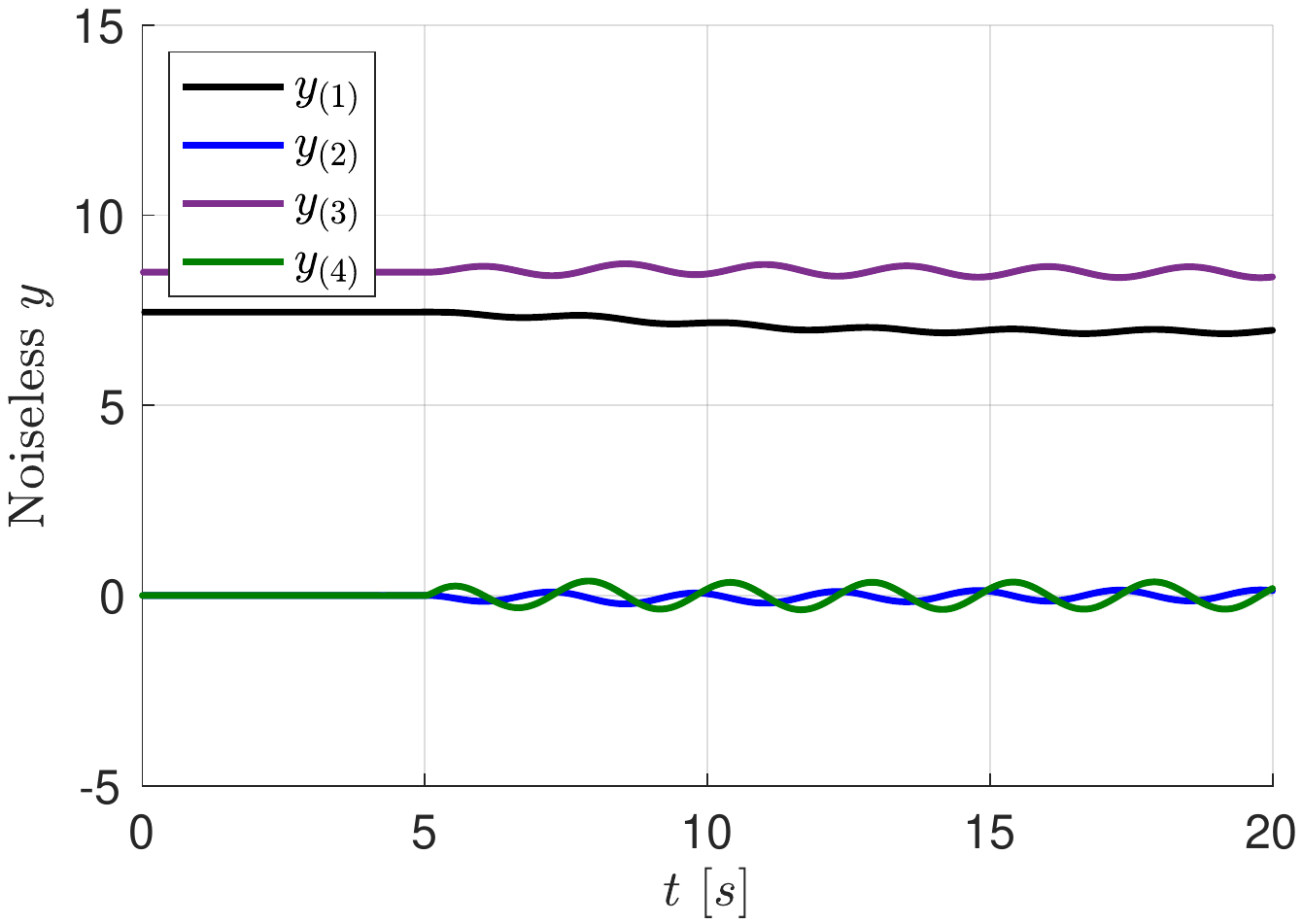}
  \end{minipage}
  \begin{minipage}{.33\textwidth}
   \includegraphics[width=1\columnwidth,height=4cm, trim={1cm 0.1cm 7cm 18.2cm},clip]{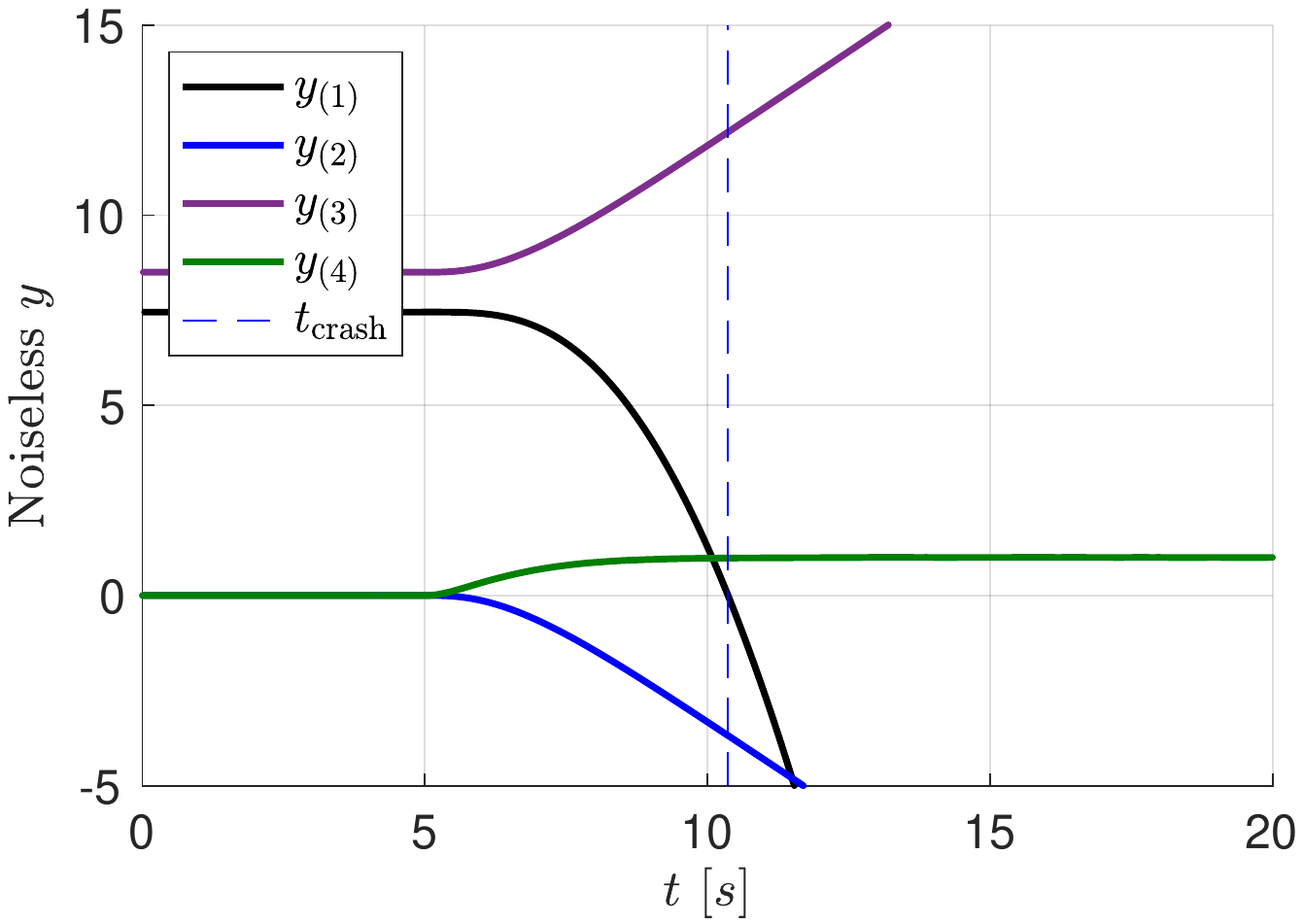}
  \end{minipage}
  \label{fig:measurements}
  \caption{The vehicle response to the applied cyber-attacks. left: Quantifiable, Middle: non-stealthy, right: Stealthy}
  \vspace{-0.6cm}
\end{figure*}
\section{Conclusion}\label{sec:conclusion}
\noindent Cyber-attacks form a great threat to the safe operation of cyber physical systems, such as autonomous vehicle platoons. By attacking communication channels or local sensors, the behaviour of the vehicles can be altered, potentially causing crashes. To address this issue, in this paper, the performance of a Sliding Mode Observer (SMO)-based detection method from previous work against malicious attacks is analysed. To this end, a a Collaborative Adaptive Cruise Control (CACC) platoon is considered where simultaneous attacks on communication and local sensors is possible.

It has been shown that this SMO-based approach can achieve detection and estimation for a meaningful class of attacks. Amongst others, it has been proven that the class of \emph{stealthy} attacks is equivalent to the class of attacks for which the attacked and healthy residual are indistinguishable. Furthermore, it has been proven that attacks are \emph{quantifiable} as long as the relative velocity measurement is not attacked.

Based on simulation scenarios with carefully crafted malicious attacks within each class, it is concluded that: quantifiable can cause significant, but limited harm to the system; stealthy attacks can cause unlimited harm to the system, but require model knowledge by the attacker; the remaining, non-stealthy, attacks cause very limited harm to the system.

In conclusion, while a large part of cyber-attacks can be estimated and/or detected by the proposed SMO-based cyber-attack detection method, it is unable to detect all possible attacks. In future research the specific structure of the stealthy attack might be used to identify specifically such attacks using other detection methods.

\bibliographystyle{IEEEtran}
\bibliography{references}

\begin{thebibliography}{10}
\providecommand{\url}[1]{#1}
\csname url@samestyle\endcsname
\providecommand{\newblock}{\relax}
\providecommand{\bibinfo}[2]{#2}
\providecommand{\BIBentrySTDinterwordspacing}{\spaceskip=0pt\relax}
\providecommand{\BIBentryALTinterwordstretchfactor}{4}
\providecommand{\BIBentryALTinterwordspacing}{\spaceskip=\fontdimen2\font plus
\BIBentryALTinterwordstretchfactor\fontdimen3\font minus
  \fontdimen4\font\relax}
\providecommand{\BIBforeignlanguage}[2]{{%
\expandafter\ifx\csname l@#1\endcsname\relax
\typeout{** WARNING: IEEEtran.bst: No hyphenation pattern has been}%
\typeout{** loaded for the language `#1'. Using the pattern for}%
\typeout{** the default language instead.}%
\else
\language=\csname l@#1\endcsname
\fi
#2}}
\providecommand{\BIBdecl}{\relax}
\BIBdecl

\bibitem{Ding2018}
D.~Ding, Q.~L. Han, Y.~Xiang, X.~Ge, and X.~M. Zhang, ``{A survey on security
  control and attack detection for industrial cyber-physical systems},''
  \emph{Neurocomputing}, vol. 275, pp. 1674--1683, 2018.

\bibitem{Dibaji2019}
S.~M. Dibaji, M.~Pirani, D.~B. Flamholz, A.~M. Annaswamy, K.~H. Johansson, and
  A.~Chakrabortty, ``{A systems and control perspective of CPS security},''
  \emph{Annual Reviews in Control}, vol.~47, pp. 394--411, 2019.

\bibitem{Sun2018}
C.~C. Sun, A.~Hahn, and C.~C. Liu, ``{Cyber security of a power grid:
  State-of-the-art},'' \emph{International Journal of Electrical Power and
  Energy Systems}, vol.~99, no. December 2017, pp. 45--56, 2018.

\bibitem{Ferrari2021-xs}
R.~M.~G. Ferrari and A.~M.~H. Teixeira, Eds., \emph{Safety, security, and
  privacy for cyber-physical systems}, ser. Lecture Notes in Control and
  Information Sciences.\hskip 1em plus 0.5em minus 0.4em\relax Springer, 2021.

\bibitem{Teixeira2015}
A.~Teixeira, I.~Shames, H.~Sandberg, and K.~H. Johansson, ``{A secure control
  framework for resource-limited adversaries},'' \emph{Automatica}, vol.~51,
  pp. 135--148, 2015.

\bibitem{Jahanshahi_2018aa}
N.~Jahanshahi and R.~M. Ferrari, ``Attack detection and estimation in
  cooperative vehicles platoons: A sliding mode observer approach,'' in
  \emph{NECSYS, Procs. of 7th IFAC Workshop on}, 2018.

\bibitem{Keijzer2019}
T.~Keijzer and R.~M.~G. Ferrari, ``A sliding mode observer approach for attack
  detection and estimation in autonomous vehicle platoons using event triggered
  communication,'' in \emph{Conf. on Decision and Control}, 2019.

\bibitem{Ploeg2011a}
J.~Ploeg, B.~T.~M. Scheepers, E.~van Nunen, N.~van~de Wouw, and H.~Nijmeijer,
  ``{Design and experimental evaluation of cooperative adaptive cruise
  control},'' in \emph{14th International IEEE Conference on Intelligent
  Transportation Systems}, 2011, pp. 260--265.

\bibitem{pasqualetti2013attack}
F.~Pasqualetti, F.~D{\"o}rfler, and F.~Bullo, ``Attack detection and
  identification in cyber-physical systems,'' \emph{IEEE transactions on
  automatic control}, vol.~58, no.~11, pp. 2715--2729, 2013.

\bibitem{Ao2016}
W.~Ao, Y.~Song, and C.~Wen, ``{Adaptive cyber-physical system attack detection
  and reconstruction with application to power systems},'' \emph{IET Control
  Theory and Applications}, vol.~10, no.~12, pp. 1458--1468, 2016.

\bibitem{Edwards2000-vs}
C.~Edwards, S.~K. Spurgeon, and R.~J. Patton, ``Sliding mode observers for
  fault detection and isolation,'' \emph{Automatica}, vol.~36, no.~4, pp.
  541--553, Apr. 2000.

\bibitem{Tan2003b}
C.~P. Tan and C.~Edwards, ``{Sliding mode observers for reconstruction of
  simultaneous actuator and sensor faults},'' \emph{Proceedings of the IEEE
  Conference on Decision and Control}, vol.~2, pp. 1455--1460, 2003.

\end{thebibliography}
\end{document}